\documentclass[runningheads]{llncs}
\usepackage[T1]{fontenc}

\usepackage{tikz}
\usepackage{listings}

\usepackage{amssymb}
\usepackage{amsmath}
\usepackage{listings}
\usepackage[utf8]{inputenc}
\usepackage{parskip}
\usepackage{hyperref}
\usepackage{bookmark}
\usepackage[linguistics]{forest}
\usepackage{tikz}
\usetikzlibrary{arrows.meta}
\usepackage{float}
\usepackage{transparent}

\usepackage{import}
\usepackage{pdfpages}
\usepackage{transparent}
\usepackage{xcolor}
\usepackage{xspace}
\usepackage{enumitem}
\usepackage{booktabs}

\newcommand{\spin}[0]{\textsc{Spin}\xspace}
\newcommand{\korg}[0]{\textsc{Tofu}\xspace}
\newcommand{\promela}[0]{\textsc{Promela}\xspace}

\usepackage{tcolorbox}
\tcbuselibrary{breakable}

\tcolorboxenvironment{definition}{
  colback=gray!8,
  colframe=gray!8,
  breakable,
  sharp corners,
  boxrule=0pt,
  left=6pt, right=6pt, top=6pt, bottom=6pt,
}

\newcommand{\ndmark}{\,\textcolor{gray!70}{\vrule width 2pt}\,}
\newenvironment{ndgroup}{%
  \par\vspace{2pt}\noindent
  \hspace{-4pt}%
  \begin{tikzpicture}[baseline=(content.north west)]
    \node[inner sep=0pt, outer sep=0pt, anchor=north west] (content) \bgroup
    \begin{minipage}{\dimexpr\linewidth-8pt}
}{%
    \end{minipage}\egroup;
    \draw[gray!60, line width=2pt]
      ([xshift=-6pt]content.north west) -- ([xshift=-6pt]content.south west);
  \end{tikzpicture}\vspace{2pt}
}

\makeatletter
\def\@spbegintheorem#1#2#3#4{\trivlist
    \item[\hskip\labelsep{#3#1\ #2\@thmcounterend}]#4\upshape}
\def\@spopargbegintheorem#1#2#3#4#5{\trivlist
      \item[\hskip\labelsep{#4#1\ #2}]{#4(#3)\@thmcounterend\ }#5\upshape}
\makeatother

\begin{document}
\title{Automated Channel Fault Analysis with \korg}
\author{Jacob Ginesin\inst{1,3} \and
Max von Hippel\inst{2} \and
Cristina Nita-Rotaru\inst{3}}
\authorrunning{J. Ginesin et al.}

\institute{Carnegie Mellon University \and
Benchify, Inc.
\and
Northeastern University
}

\maketitle

\begin{abstract}
Distributed protocols are the linchpin of the modern internet,
underpinning every internet service.
This has in turn motivated a massive body of research ensuring the security,
reliability, and performance of distributed protocols.
In these works, a wide-ranging assumption is that distributed
protocols operate over \textit{faulty} or
\textit{attacker-controlled} channels, where messages can be arbitrarily
inserted, dropped, replayed, or reordered.
Formal verification work targeting distributed protocols
typically defines its own notion of faulty or malicious channels,
then constructively proves their protocol is correct with respect to it.
In this work we take a fundamentally different approach:
we develop a rigorous methodology for
automatically conducting channel fault analysis
on distributed protocols, and we introduce \korg, a 
generalizable tool that implements our methodology. 
\korg provides sound, complete analysis, synthesizing channel fault-based 
attack traces on arbitrary
linear temporal logic (LTL) protocol specifications
or proving the absence of such through an exhaustive state-space search.
We demonstrate the applicability of \korg by employing it to study TCP.

\end{abstract}

\section{Introduction}
\label{sec:intro}

Distributed protocols are the linchpin of the modern internet,
representing the fundamental communication and coordination backbone
of all modern services built on top of it.
The vast importance of distributed protocols has motivated ample
research in ensuring their security, reliability, and performance.

In practice, the communication channels between distributed protocol
endpoints are imperfect: messages can be dropped due to
network congestion, replayed by adversaries, or reordered due to
multi-path routing. These channel-level faults are a source
of subtle bugs, denial-of-service vulnerabilities, and
safety violations in distributed systems.
Despite this, the landscape of tools for \textit{systematically}
analyzing channel faults remains fragmented and largely manual.

On the practical side, several tools exist for testing network
protocol implementations under adverse channels.
Scapy \cite{tools/scapy} 
is a Python-based packet manipulation library enabling
users to forge, send, and capture arbitrary network packets, Linux 
\texttt{netem} \cite{tools/netem} provides kernel-level network emulation, 
PacketDrill \cite{tools/packetdrill} supports testing of TCP/UDP/IP implementations
through injecting precisely-timed packet sequences, and TCP-Fuzz supports randomized testing
and fuzzing approaches.

These tools are invaluable in practice but share a fundamental
limitation: they provide no \textit{completeness} guarantees.
That is, when these tools terminate without finding a bug,
the user cannot conclude the absence of bugs.

On the formal methods side, there has been significant work verifying
distributed protocols under faulty communication channels, with the caveat that 
the capabilities of the attacker must be manually defined. von Hippel et al.
presented an automated attack synthesis framework targeting communication channels
for distributed systems, supporting an attacker that may only arbitrarily insert 
messages onto the victim channel \cite{conf/safecomp/Hippel20}. This attack synthesis framework was later used to analyze the
Stream Control Transmission Protocol, where 
Ginesin et al. modeled different types of attackers by specifying exact message types the different attackers may insert onto the channel \cite{conf/usenix/ginesin24}. 
Raft has been verified using mechanized proof assistants such as Rocq via the Verdi framework
\cite{conf/ccs/Wilcox15}, and Paxos has been verified using IronFleet \cite{conf/sosp/Hawblitzel15};
in both scenarios, ad-hoc presentations of channel faults are verified against.
Cryptographic protocol verifiers such as Tamarin \cite{tools/tamarin} and ProVerif \cite{tools/proverif}
provide automated analysis of cryptographic properties when processes communicate over an untrusted channel, but do not target channel-level faults. Model checkers such as \spin \cite{tools/spin} and TLA+ \cite{tools/tla} have been applied
to distributed protocols, but there exists very little
systematic support for reasoning directly about channel faults.

As it stands, the current formal methods-driven approaches 
for protocol validation lack support for automated channel
fault analysis, while testing-driven approaches remain 
fundamentally incomplete. To fill in this gap in the literature, 
we propose a general methodology for the automated discovery of
attack traces, counterexamples, and formal guarantees across various 
channel configurations. To realize our methodology we introduce \korg,
a highly generalizable tool for synthesizing attacks on the faulty 
communication channels of distributed protocols that supports attacker 
models capable of arbitrarily dropping, replaying, reordering messages.
\korg targets the communication channels between protocol endpoints
and synthesizes attacks to violate general linear temporal logic (LTL)
specifications.
\korg is designed to either synthesize an attack or prove
the absence of such via an exhaustive state-space search.
\korg is sound and complete: if there exists an attack, \korg will
find it, and \korg will never produce false positives.
\korg is also designed to be easy to use once the protocol model
is constructed: all a user must do is select the victim channels,
select the desired attacker models, and invoke \korg.

We summarize our contributions:
\begin{itemize}
\item We present generalizable gadgets representing attackers
      capable of dropping, replaying, and reordering messages
      on communication channels of distributed protocols,
      and we formally describe their construction.
\item We present \korg, a tool for synthesizing attacks against
      distributed communication protocols.
      \korg supports three general attacker models: an attacker
      that can drop, replay, or reorder messages on a channel.
\item We provide an overview of \korg and demonstrate its usage
      by walking through applying it to the transmission control protocol (TCP).
\end{itemize}

We release our code and models as open source at \url{https://github.com/JakeGinesin/tofu}.

\section{Preliminaries}
\label{sec:prelim}

In this section we introduce the mathematical foundations
underpinning \korg. We define linear temporal logic and B\"uchi
Automata, formalize
distributed protocols as compositions of communicating processes,
and state the attack synthesis problem.

\subsection{Linear Temporal Logic and the B\"uchi Automata}
\label{sub:ltl}

Linear Temporal Logic (LTL) \cite{conf/focs/pnueli77} is a modal logic for specifying
properties of infinite sequences of program states.
LTL extends propositional logic with temporal operators that
express how propositions evolve over time along an execution.

\begin{definition}[LTL Syntax]
\label{def:ltl-syntax}
Given a finite set of atomic propositions $AP$, LTL formulas
are defined inductively:
\[
\phi ::= \top \mid p \mid \neg \phi \mid \phi_1 \land \phi_2
       \mid \mathbf{X}\, \phi \mid \phi_1 \,\mathbf{U}\, \phi_2
\]
where $p \in AP$, $\mathbf{X}$ (\textit{next}) asserts $\phi$ holds
in the next state, and $\mathbf{U}$ (\textit{until}) asserts
$\phi_1$ holds continuously until $\phi_2$ holds.

The standard LTL derived operators are defined as:
$\mathbf{F}\,\phi \equiv \top\,\mathbf{U}\,\phi$
(\textit{eventually}, asserting $\phi$ holds at some future state),
$\mathbf{G}\,\phi \equiv \neg\mathbf{F}\,\neg\phi$
(\textit{always}, asserting $\phi$ holds at every future state),
and $\phi_1 \Rightarrow \phi_2 \equiv \neg\phi_1 \lor \phi_2$
(\textit{implies}).
\end{definition}

LTL properties are commonly classified into \textit{safety}
and \textit{liveness}. Safety properties assert
that ``a bad thing never happens'' and can be violated by
finite prefixes, while liveness properties assert that
``a good thing eventually happens'' and require infinite
traces to be violated. This distinction is important for
\korg's gadget design, as we discuss in
Section~\ref{sec:gadgets}.

We say a program $P$ \textit{models} a property $\phi$,
written $P \models \phi$, if $\phi$ holds over every infinite
execution trace of $P$. LTL model checking, the problem of
deciding whether $P \models \phi$ for a finite-state $P$, 
is decidable, and model checkers will always terminate given
sufficient resources.

LTL is fundamentally equivalent to the B\"uchi
Automaton, which expresses languages over infinite words. B\"uchi 
Automata are defined as follows.

\begin{definition}[B\"uchi Automaton]
\label{def:ba}
A \textit{B\"uchi automaton} (BA) is a tuple
$B = (Q, \Sigma, \delta, Q_0, F)$ where:
\begin{itemize}
  \item $Q$ is a finite set of states,
  \item $\Sigma$ is a finite alphabet,
  \item $\delta \subseteq Q \times \Sigma \times Q$ is a
        transition relation,
  \item $Q_0 \subseteq Q$ is a set of initial states,
  \item $F \subseteq Q$ is a set of \textit{accepting} states.
\end{itemize}
A \textit{run} of $B$ on an infinite word
$w = a_0 a_1 a_2 \cdots \in \Sigma^\omega$ is an infinite
sequence of states $q_0, q_1, q_2, \ldots$ such that
$q_0 \in Q_0$ and $(q_i, a_i, q_{i+1}) \in \delta$ for all
$i \geq 0$. The run is \textit{accepting} if it visits
states in $F$ infinitely often. The \textit{language} of $B$,
denoted $\mathcal{L}(B)$, is the set of all infinite words
for which $B$ has an accepting run.
\end{definition}

Model checkers such as \spin are fundamentally reliant on the
equivalence of LTL and B\"uchi Automata to implement LTL model checking \cite{tools/spin}. 
In LTL model checking, a program is modeled as a B\"uchi Automata, and the 
LTL property of interest is negated then converted to another B\"uchi Automata -- finally, both automata 
are intersected and exhaustively searched for accepting runs \cite{conf/lics/Vardi86}.

The execution semantics, as well as the formal equivalence relation between LTL and B\"uchi Automata,
are elaborated upon in Appendix Section \ref{sec:LTL-BA-SEC}.

\subsection{Processes and Channels}
\label{sub:processes}

We formalize distributed protocol components as
\textit{processes} --- labeled transition systems augmented
with input/output distinctions on transitions, following
the distributed protocol synthesis framework of Alur and Tripakis \cite{acm/Alur17}.

\begin{definition}[Process]
\label{def:process}
A \textit{process} is a tuple
$P = \langle AP, I, O, S, s_0, T, L \rangle$ where
$AP$ is a finite set of atomic propositions;
$I$ and $O$ are finite, disjoint sets of \textit{input} and \textit{output} actions, respectively;
$S$ is a finite set of states with initial state $s_0 \in S$;
$T \subseteq S \times (I \cup O \cup \{\tau\}) \times S$ is the transition relation, where $\tau \notin I \cup O$ denotes a silent internal action; and
$L : S \to 2^{AP}$ is a labeling function mapping each state to the set of atomic propositions true in that state.
A transition $(s, x, s') \in T$ is an \textit{input transition}
if $x \in I$, an \textit{output transition} if $x \in O$,
and an \textit{internal transition} if $x = \tau$.
\end{definition}

Communication between processes occurs over \textit{channels}.
Channels mediate all inter-process communication by providing
buffered, typed message conduits.

\begin{definition}[Channel Process]
\label{def:channel}
A \textit{channel} is a tuple $c = (M, k)$ where $M$ is a
finite \textit{message domain} (the set of messages that
can be transmitted over $c$) and $k \in \mathbb{N}$ is the
\textit{buffer capacity}. The state of a channel at any
point in an execution is a sequence $\mathbf{b} \in M^{\leq k}$
(i.e., a sequence over $M$ of length at most $k$),
representing the contents of the channel's FIFO buffer.

Channels induce specific input and output actions on the
processes that use them. For a channel $c = (M, k)$, we
define the associated action sets:
 \begin{align*}
  \mathsf{Send}(c) &= \{ c\, ! \, m \mid m \in M \}  \\
  \mathsf{Recv}(c) &= \{ c \,? \, m \mid m \in M \}   \\
  \mathsf{Peek}(c) &= \{ c \, ? \, \langle m \rangle \mid m \in M \}
\end{align*}
where $c\,!\,m$ denotes placing message $m$ onto the tail of
$c$'s buffer (enabled when $|\mathbf{b}| < k$),
$c\,?\,m$ denotes removing message $m$ from the head of
$c$'s buffer (enabled when $\mathsf{head}(\mathbf{b}) = m$),
and $c\,?\,\langle m \rangle$ denotes \textit{non-destructively}
reading message $m$ from the head of $c$'s buffer without
removing it (enabled when $\mathsf{head}(\mathbf{b}) = m$).
For a sending process, $\mathsf{Send}(c) \subseteq O$;
for a receiving process, $\mathsf{Recv}(c) \subseteq I$.
The non-destructive read $\mathsf{Peek}(c)$ is an input
action that does not modify the channel state.

\end{definition}

In general, a \textit{distributed protocol model} consists of a finite
collection of processes $P_1, \ldots, P_n$ communicating
over a finite set of channels $c_1, \ldots, c_r$. To facilitate this, we define
a notion of \textit{asynchronous composition} of processes 
that models concurrent execution.

\begin{definition}[Asynchronous Composition]
\label{def:composition}
Given processes $P = \langle AP_P, I_P, O_P, S_P, s_0^P, T_P, L_P \rangle$
and $Q = \langle AP_Q, I_Q, O_Q, S_Q, s_0^Q, T_Q, L_Q \rangle$
such that $O_P \cap O_Q = \emptyset$ and $AP_P \cap AP_Q = \emptyset$,
their \textit{asynchronous composition} $P \parallel Q$ is the process
$\langle AP_P \cup AP_Q,\;
(I_P \cup I_Q) \setminus (O_P \cup O_Q),\;
O_P \cup O_Q,\;
S_P \times S_Q,\;
(s_0^P, s_0^Q),\;
T_\parallel,\;
L_\parallel \rangle$
where $L_\parallel(s_P, s_Q) = L_P(s_P) \cup L_Q(s_Q)$ and
$T_\parallel$ consists of:
\begin{itemize}[nosep,leftmargin=1.5em]
  \item \textbf{Left move.}
    $((s_P, s_Q),\; \alpha,\; (s_P', s_Q)) \in T_\parallel$
    if $(s_P, \alpha, s_P') \in T_P$
    and $\alpha \notin I_Q \cup O_Q$.
  \item \textbf{Right move.}
    $((s_P, s_Q),\; \alpha,\; (s_P, s_Q')) \in T_\parallel$
    if $(s_Q, \alpha, s_Q') \in T_Q$
    and $\alpha \notin I_P \cup O_P$.
  \item \textbf{Synchronization.}
    $((s_P, s_Q),\; \alpha,\; (s_P', s_Q')) \in T_\parallel$
    if $(s_P, \alpha, s_P') \in T_P$ and $(s_Q, \alpha, s_Q') \in T_Q$
    for $\alpha \in (O_P \cap I_Q) \cup (O_Q \cap I_P)$.
\end{itemize}
Since $\tau \notin I \cup O$ for any process (Definition~\ref{def:process}),
internal transitions are always local moves: a $\tau$-transition of $P$
satisfies the left-move condition (as $\tau \notin I_Q \cup O_Q$),
and symmetrically for $Q$.
Composition generalizes to $n$ processes:
$P_1 \parallel \cdots \parallel P_n$ communicating over
channels $c_1, \ldots, c_r$, where each process's I/O
actions are drawn from the channel action sets of
Definition~\ref{def:channel}. We call such a collection a
\textit{distributed protocol model}.
\end{definition}

Nondeterminism in the composite system arises directly from
$T_\parallel$: whenever multiple transitions are enabled in a
composite state $(s_P, s_Q)$, the system may take any of them.
In particular, the nondeterministic choices within the gadgets
of Section~\ref{sec:gadgets} are instances of
multiple enabled transitions from a single state,
and model checkers will exhaustively explore all such interleavings.

\subsection{Process--B\"uchi Automaton Correspondence}
\label{sub:proc-ba}

A key technical requirement for \korg is establishing that
the process formalism corresponds precisely to the B\"uchi
automata semantics used by the \spin model checker.

\begin{theorem}
\label{thm:proc-ba}
Every process $P = \langle AP, I, O, S, s_0, T, L \rangle$
directly corresponds to a B\"uchi automaton, and vice versa.
\end{theorem}
\vspace{-1.25em}
\begin{proof}
Provided in Appendix Section \ref{sub:expanded-proofs}.    
\end{proof}
\vspace{-0.7em}

This correspondence ensures that reasoning about processes
within the attack synthesis framework directly reduces
to B\"uchi Automata operations in \spin in sub-polynomial time.

\subsection{Attack Synthesis}
\label{sub:attack-synth}

\textit{LTL attack synthesis} \cite{conf/safecomp/Hippel20} is the dual of LTL program synthesis.
In program synthesis, given an LTL specification $\phi$,
the goal is to derive a program $P$ such that $P \models \phi$.
In attack synthesis, the problem is inverted: given a system
$P$ and a property $\phi$ such that $P \models \phi$, we ask
whether there exists some process $A$, the \textit{attacker} process,
such that $(P \parallel A) \not\models \phi$.

\begin{definition}[Threat Model]
\label{def:threat}
A \textit{threat model} is a tuple $(P, (Q_i)_{i=0}^m, \phi)$ where:
\begin{itemize}
  \item $P, Q_0, \ldots, Q_m$ are processes,
  \item each $Q_i$ has $AP = \emptyset$,
  \item $\phi$ is an LTL formula such that
        $P \parallel Q_0 \parallel \cdots \parallel Q_m \models \phi$,
  \item the composite system has at least one infinite run
        (non-triviality).
\end{itemize}
\end{definition}

The \textit{attack synthesis problem} (R-$\exists$ASP) asks:
given a threat model $(P, (Q_i)_{i=0}^m, \phi)$, does there
exist an attacker $A$ whose I/O interface is a subset of
that of the $Q_i$ such that
$P \parallel \text{Daisy}(Q_0) \parallel \cdots
\parallel \text{Daisy}(Q_m) \not\models \phi$?
Here, $\text{Daisy}(Q_i)$ is a gadget that nondeterministically
exercises the full I/O interface of $Q_i$, as defined by
von Hippel et al. \cite{conf/safecomp/Hippel20}.

In the general attack synthesis framework, the $Q_i$ are
arbitrary environment processes; however, in \korg, the $Q_i$
are instantiated as the channel processes mediating
communication between the protocol endpoints in $P$.
That is, each $Q_i$ models a communication channel
(Definition~\ref{def:channel}), and \korg synthesizes
attackers that replace one or more $Q_i$ with fault
gadgets exercising the same I/O interface. 
Given a user-selected subset of victim channels
$\{Q_{i_1}, \ldots, Q_{i_j}\} \subseteq \{Q_0, \ldots, Q_m\}$, and 
letting $\text{Gadget}(Q_j)$ denote the fault gadget
(Section~\ref{sec:gadgets}) synthesized by \korg for
channel $Q_j$, \korg then checks whether:
\[
  P \parallel \text{Gadget}(Q_0) \parallel \cdots
  \parallel \text{Gadget}(Q_j) \not\models \phi
\]
Fundamentally, \korg is a synthesizer for such a $\text{Gadget}(Q_j)$; 
model checking is then leveraged to check whether $P \parallel \text{Gadget}(Q_j) \models \phi$.
Additionally, using the process-B\"uchi Automata correspondence as described in Theorem~\ref{thm:proc-ba}, we can 
establish a complexity result for the $R-\exists$ASP problem, which \korg inherits:

\begin{theorem}
\label{thm:easp-pspace}
The R-$\exists$ASP problem is reducible to B\"uchi Automata
language inclusion, and is therefore in PSPACE.
\end{theorem}

\begin{proof}[Proof sketch]
By Theorem~\ref{thm:proc-ba}, all processes in the threat model
correspond to B\"uchi Automata. Checking
$P \parallel \text{Daisy}(Q_0) \parallel \cdots
\parallel \text{Daisy}(Q_m) \models \phi$
reduces to checking whether
$\mathcal{L}(BA_P \parallel BA_{\text{Daisy}(Q_0)} \parallel
\cdots \parallel BA_{\text{Daisy}(Q_m)}) \subseteq
\mathcal{L}(BA_\phi)$,
which is B\"uchi Automata language inclusion, a problem well-known to be
PSPACE-complete \cite{journal/tcs/sistla87}. 
\end{proof}

\begin{remark}
  It is always the case that for a gadget \text{Gadget}($Q$) 
  synthesized by \korg, $\text{Gadget}(Q) \subseteq \text{Daisy}(Q)$.
  Therefore, all model checking operations dispatched by \korg 
  have a complexity of PSPACE. 

\end{remark}

\section{Channel Fault Gadgets}
\label{sec:gadgets}

In this section, we introduce our methodology for the automatic
synthesis of drop, replay, and reorder attacks on distributed
protocols. We design our gadgets in \textit{generality}: they
can be applied to arbitrary communication channels, agnostic
to the message types actually transmitted. These gadgets are
designed to be composed with the protocol processes under
analysis, following the gadgetry-based attack synthesis
framework of von Hippel et al. \cite{conf/safecomp/Hippel20}.

We now define the drop, replay, and reorder gadgets as
processes (Definition~\ref{def:process}). Each gadget targets
a victim channel $c = (M, k)$ with a user-specified fault
limit $\ell \in \mathbb{N}$.
We highlight sets of transitions that may be
exercised \emph{nondeterministically} from the same state
with a vertical bar on the left margin (\ndmark{}).
For clarity, we write transitions as arrows
$s \xrightarrow{\alpha} s'$ rather than as three-tuples
$(s, \alpha, s')$, where $s$ is the source state, $\alpha$
is the action label, and $s'$ is the target state.

\subsection{Drop Gadget}
\label{sub:drop-gadget}

The drop gadget silently consumes up to $\ell$ messages from $c$.

\begin{definition}[Drop Gadget]
\label{def:drop}
The drop gadget $D_\ell(c)$ is a process
$\langle AP, I, O, S, s_0, T, L \rangle$ where:
\begin{itemize}[nosep,leftmargin=1.5em]
  \item $AP = \{\mathit{done}\}$
  \item $I = \mathsf{Recv}(c) \cup \{\mathsf{skip}(c)\}$
  \item $O = \emptyset$
  \item $S = (\{\textsc{Main}\} \times \{0, \ldots, \ell\})
              \;\cup\; \{\textsc{End}\}$
  \item $s_0 = (\textsc{Main}, \ell)$
  \item $L(s) = \{\mathit{done}\}$ iff $s = \textsc{End}$,
        $\emptyset$ otherwise
\end{itemize}

$T$ consists of the following transitions for all $m \in M$:

\smallskip\noindent\textbf{Main Phase:}\vspace{-1.0\baselineskip}
\begin{ndgroup}
\begin{enumerate}
  \item \textbf{Drop.}
    $(\textsc{Main}, n)
    \xrightarrow{c\,?\,m}
    (\textsc{Main}, n{-}1)$
    for $n > 0$.
  \item \textbf{Pass.}
    $(\textsc{Main}, n)
    \xrightarrow{\mathsf{skip}(c)}
    (\textsc{Main}, n)$.
  \item \textbf{Terminate.}
    $(\textsc{Main}, n)
    \xrightarrow{\tau}
    \textsc{End}$.
\end{enumerate}
\end{ndgroup}

\smallskip\noindent\textbf{Break Phase:}
\begin{enumerate}
  \setcounter{enumi}{3}
  \item \textbf{End.}
    $\textsc{End}
    \xrightarrow{\tau}
    \textsc{End}$.
\end{enumerate}
\end{definition}

The model checker explores all dropping strategies up to $\ell$
via the nondeterministic choice between \textbf{Drop},
\textbf{Pass}, and \textbf{Terminate}. \textbf{Pass} prevents
the gadget from trivially violating liveness by starving the
recipient.

\subsection{Replay Gadget}
\label{sub:replay-gadget}

The replay gadget observes messages on $c$ via non-destructive
reads and later replays copies onto $c$.

\begin{definition}[Replay Gadget]
\label{def:replay}
The replay gadget $R_\ell(c)$ is a process
$\langle AP, I, O, S, s_0, T, L \rangle$ where:
\begin{itemize}[nosep,leftmargin=1.5em]
  \item $AP = \{\mathit{done}\}$
  \item $I = \mathsf{Peek}(c) \cup \{\mathsf{skip}(c)\}$
  \item $O = \mathsf{Send}(c)$
  \item $S = (\{\textsc{Consume}\} \times \{0, \ldots, \ell\}
              \times \mathsf{Buf}_\ell(M))
         \;\cup\;
              (\{\textsc{Replay}\} \times \{0, \ldots, \ell\}
              \times \mathsf{Buf}_\ell(M))
         \;\cup\; \{\textsc{End}\}$
  \item $s_0 = (\textsc{Consume}, \ell, \varepsilon)$
  \item $L(s) = \{\mathit{done}\}$ iff $s = \textsc{End}$,
        $\emptyset$ otherwise
\end{itemize}

$T$ consists of the following transitions for all $m \in M$:

\smallskip\noindent\textbf{Consume Phase:}\vspace{-0.75\baselineskip}
\begin{ndgroup}
\begin{enumerate}[nosep,leftmargin=1.5em]
  \item \textbf{Observe.}
    $(\textsc{Consume}, n, b)
    \xrightarrow{c\,?\,\langle m \rangle}
    (\textsc{Consume}, n{-}1, b \cdot m)$
    for $n > 1$, $|b| < \ell$.
  \item \textbf{Observe-then-replay.}
    $(\textsc{Consume}, n, b)
    \xrightarrow{c\,?\,\langle m \rangle}
    (\textsc{Replay}, \ell, b \cdot m)$
    for $n \geq 1$, $|b| < \ell$.
  \item \textbf{Begin-Replay.}
    $(\textsc{Consume}, 0, b)
    \xrightarrow{\tau}
    (\textsc{Replay}, \ell, b)$
    for $|b| < \ell$.
  \item \textbf{Pass.}
    $(\textsc{Consume}, n, b)
    \xrightarrow{\mathsf{skip}(c)}
    (\textsc{Consume}, n, b)$.
\end{enumerate}
\end{ndgroup}

\smallskip\noindent\textbf{Replay Phase:}\vspace{-0.75\baselineskip}
\begin{ndgroup}
\begin{enumerate}[nosep,leftmargin=1.5em]
  \setcounter{enumi}{4}
  \item \textbf{Replay.}
    $(\textsc{Replay}, n, b)
    \xrightarrow{c\,!\,m}
    (\textsc{Replay}, n{-}1, b \setminus m)$
    for $m \in b$.
  \item \textbf{Replay-No-Delete.}
    $(\textsc{Replay}, n, b)
    \xrightarrow{c\,!\,m}
    (\textsc{Replay}, n, b)$
    for $m \in b$.
  \item \textbf{Discard.}
    $(\textsc{Replay}, n, b)
    \xrightarrow{\tau}
    (\textsc{Replay}, n, b \setminus m)$
    for $m \in b$.
  \item \textbf{Pass.}
    $(\textsc{Replay}, n, b)
    \xrightarrow{\mathsf{skip}(c)}
    (\textsc{Replay}, n, b)$.
  \item \textbf{Terminate.}
    $(\textsc{Replay}, n, b)
    \xrightarrow{\tau}
    \textsc{End}$.
  \item \textbf{Empty-$\varepsilon$.}
    $(\textsc{Replay}, n, \varepsilon)
    \xrightarrow{\tau}
    \textsc{End}$
  \item \textbf{Empty-$\ell$.}
    $(\textsc{Replay}, 0, b)
    \xrightarrow{\tau}
    \textsc{End}$.
\end{enumerate}
\end{ndgroup}

\smallskip\noindent\textbf{Break Phase:}
\begin{enumerate}[nosep,leftmargin=1.5em]
  \setcounter{enumi}{11}
  \item \textbf{End.}
    $\textsc{End}
    \xrightarrow{\tau}
    \textsc{End}$.
\end{enumerate}
\end{definition}

Transitions~1--2 in the \textbf{Consume Phase} 
use the non-destructive read
$c\,?\,\langle m \rangle$: messages are copied without removal
from the channel. When $n = 1$, only \textbf{Observe-then-replay}
applies, forcing the gadget into the \textbf{Replay Phase}. When $n > 1$,
the gadget nondeterministically chooses to continue observing
or begin replaying. 
Once in the replay phase, the gadget nondeterministically chooses to \textbf{Replay} a message 
from memory then removing it, \textbf{Replay-No-Delete} a random message from memory but without deleting it, 
discarding a message from memory, or moving to the \textbf{Break Phase} nondeterministically or once
the fault limit is reached or the memory is depleted.

\textbf{Pass} prevents the gadget from trivially violating 
liveness by starving the recipient; \textbf{Terminate} allows for premature termination in any state.

\subsection{Reorder Gadget}
\label{sub:reorder-gadget}

The reorder gadget intercepts messages from $c$ via
destructive reads and replays them in a potentially
different order.

\begin{definition}[Reorder Gadget]
\label{def:reorder}
The reorder gadget $O_\ell(c)$ is a process
$\langle AP, I, O, S, s_0, T, L \rangle$ where:
\begin{itemize}[nosep,leftmargin=1.5em]
  \item $AP = \{\mathit{done}\}$
  \item $I = \mathsf{Recv}(c) \cup \{\mathsf{skip}(c)\}$
  \item $O = \mathsf{Send}(c)$
  \item $S = \{\textsc{Init}\}
         \;\cup\;
              (\{\textsc{Consume}\} \times \{0, \ldots, \ell\}
              \times \mathsf{Buf}_\ell(M))
         \;\cup\;
              (\{\textsc{Replay}\}
              \times \mathsf{Buf}_\ell(M))
         \;\cup\; \{\textsc{End}\}$
  \item $s_0 = \textsc{Init}$
  \item $L(s) = \{\mathit{done}\}$ iff $s = \textsc{End}$,
        $\emptyset$ otherwise
\end{itemize}

$T$ consists of the following transitions for all $m \in M$:

\smallskip\noindent\textbf{Init Phase:}\vspace{-1.0\baselineskip}
\begin{ndgroup}
\begin{enumerate}[nosep,leftmargin=1.5em]
  \item \textbf{Pass.}
    $\textsc{Init}
    \xrightarrow{\mathsf{skip}(c)}
    \textsc{Init}$.
  \item \textbf{Begin.}
    $\textsc{Init}
    \xrightarrow{\tau}
    (\textsc{Consume}, \ell, \varepsilon)$.
\end{enumerate}
\end{ndgroup}

\smallskip\noindent\textbf{Consume Phase:}
\begin{enumerate}[nosep,leftmargin=1.5em]
  \setcounter{enumi}{2}
  \item \textbf{Consume.}
    $(\textsc{Consume}, n, b)
    \xrightarrow{c\,?\,m}
    (\textsc{Consume}, n{-}1, b \cdot m)$
    for $n > 1$, $|b| < \ell$.
  \item \textbf{Consume-Last.}
    $(\textsc{Consume}, 1, b)
    \xrightarrow{c\,?\,m}
    (\textsc{Replay}, b \cdot m)$
    for $|b| < \ell$.
\end{enumerate}

\smallskip\noindent\textbf{Replay Phase:}
\begin{enumerate}[nosep,leftmargin=1.5em]
  \setcounter{enumi}{4}
  \item \textbf{Replay.}
    $(\textsc{Replay}, b)
    \xrightarrow{c\,!\,m}
    (\textsc{Replay}, b \setminus m)$
    for $m \in b$.
  \item \textbf{Empty.}
    $(\textsc{Replay}, \varepsilon)
    \xrightarrow{\tau}
    \textsc{End}$.
\end{enumerate}

\smallskip\noindent\textbf{Break Phase:}
\begin{enumerate}[nosep,leftmargin=1.5em]
  \setcounter{enumi}{6}
  \item \textbf{End.}
    $\textsc{End}
    \xrightarrow{\tau}
    \textsc{End}$.
\end{enumerate}
\end{definition}

The \textbf{Init Phase} nondeterministically skips an
arbitrary prefix of channel traffic via \textbf{Pass} before
entering the \textbf{Consume Phase}. Unlike the replay gadget,
consumption uses the destructive read $c\,?\,m$ and is
strictly sequential: the gadget consumes exactly $\ell$
messages before transitioning to replay. \textbf{Consume}
applies while the remaining budget is $n > 1$;
\textbf{Consume-Last} fires at $n = 1$, forcing the gadget
into the \textbf{Replay Phase}.

Once in the \textbf{Replay Phase}, \textbf{Replay}
nondeterministically selects $m \in b$ and sends it on $c$,
removing it from the buffer. This ensures all $\ell!$ orderings
are explored by the model checker. The gadget terminates via
\textbf{Empty} once the buffer is depleted, guaranteeing that
every intercepted message is delivered.

The reorder gadget executes at the highest priority in the
composite system, ensuring it intercepts messages on $c$
before other processes consume them.

\subsection{Properties of the Gadget Constructions}
\label{sub:gadget-properties}

We state four properties of the gadget constructions that
are essential for the correctness of \korg.

\begin{theorem}[Finiteness]
\label{prop:finite}
For any finite message domain $M$ and fault limit $\ell$,
the state sets of $D_\ell(c)$, $R_\ell(c)$, and $O_\ell(c)$
are finite.
\end{theorem}

\begin{proof}
The state set of each gadget is a product of a finite set of
control locations, a bounded counter in $\{0, \ldots, \ell\}$,
and a buffer drawn from $\mathsf{Buf}_\ell(M)$, which has
$\sum_{i=0}^{\ell} |M|^i$ elements. All components are finite.
\end{proof}

Finiteness is what enables a model checker to perform an exhaustive
state-space search over the composition of the protocol model
with the gadgets. Increasing $\ell$ enlarges the state space
factorially; in practice, small values of $\ell$ (typically 2--3)
are often sufficient to discover attacks while keeping
verification tractable.

\begin{theorem}[Soundness and Completeness]
\label{prop:sound-complete}
\korg provides sound and complete channel fault analysis.
\end{theorem}

\begin{proof}
By Proposition~\ref{prop:finite}, composing the protocol with any
gadget yields a finite-state system, so \spin's exploration terminates.
By Theorem~\ref{thm:easp-pspace}, verification reduces to B\"uchi
Automata language inclusion, which \spin decides exactly.
Completeness follows from exhaustive state-space search over the
finite product; soundness follows because any counterexample is a
genuine accepting run.
\end{proof}

\section{\korg Architecture}
\label{sec:design}

In this section we discuss the design and implementation
of \korg.

\subsection{High-Level Design}
\label{sub:high-level}

\korg is designed to synthesize attacks with respect to
imperfect channels. That is, \korg synthesizes attacks that
involve replaying, dropping, or reordering messages on one
or more communication channels relied upon by the victim protocol.

\korg targets user-specified communication channels in
state machine-based formal models of distributed protocols.
To use \korg, the user provides: (1) a formal model of a
distributed protocol in the \promela language, (2) the
communication channel(s) to attack, (3) the desired attacker
model (drop, replay, reorder), and (4) a
formalized correctness property in LTL. The protocol model
should satisfy the correctness property in the absence of any
attacker.

Once invoked, \korg parses the \promela model, identifies
the target channels and their message types, automatically
synthesizes the appropriate attacker gadgets, and injects
them into the model. The modified model is then passed to
\spin, which performs an exhaustive state-space search.
If an attack exists, \spin returns a counterexample trace;
otherwise, \spin certifies the absence of an attack via
complete state-space exploration.

\subsection{The \spin Model Checker}
\label{sub:spin}

We implement \korg on top of \spin, a model checker for
reasoning about distributed and concurrent systems \cite{tools/spin}. 
Models in \spin are written in \promela, a modeling language
for communicating state machines. \promela processes
communicate via built-in \textit{channel} structures, which can be either
synchronous (buffer size 0) or asynchronous (buffer size $> 0$).
\korg generates attacker gadgets \textit{with respect to}
these channels.

Internally, \spin reduces verification problems to deciding
B\"uchi Automata intersection emptiness via the
LTL--BA correspondence (Theorem~\ref{thm:ltl-ba}).
Given $n$ B\"uchi Automata $P_1, \ldots, P_n$ representing
the protocol processes and the negation of the LTL property,
\spin checks:
\[
  \mathcal{L}(P_1) \cap \mathcal{L}(P_2) \cap \cdots
  \cap \mathcal{L}(P_n) \stackrel{?}{=} \emptyset
\]
Rather than constructing the full product automaton
(which is $|P_1| \times \cdots \times |P_n|$ in size),
\spin employs \textit{on-the-fly composition}, exploring
the product state space without explicitly materializing it.

\korg could alternatively be built on top of other automated
reasoning tools, including TLA+ \cite{blog/wayne18} and Dafny \cite{journal/ifm/dione}. We choose \spin 
due to its built-in support for channels compatible 
with the interface presented in Definition \ref{def:channel}, as well as
its historical popularity and robustness.

A high-level visual overview of the \korg pipeline in relation to \spin is given in Figure \ref{fig:korg_workflow}.

\vspace{-10pt}
\begin{figure}[h]
    \centering
    \includegraphics[width=0.65\textwidth]{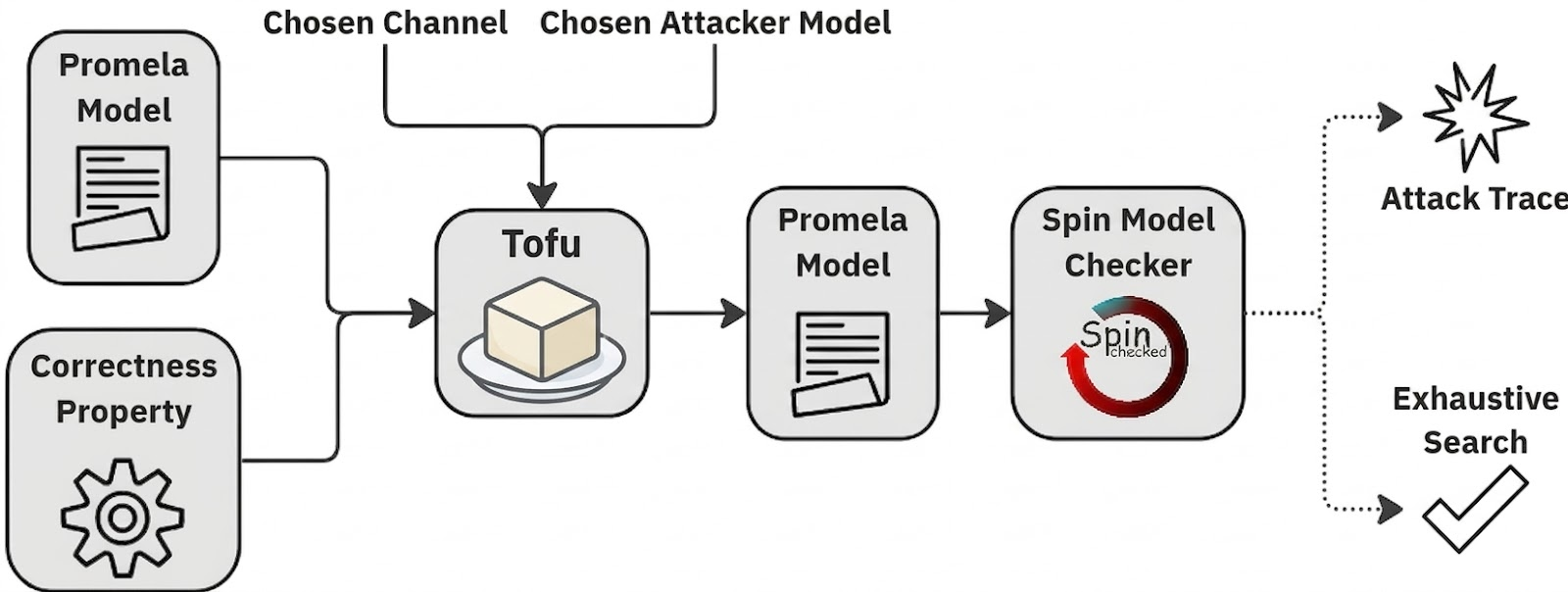}
    \caption{A high-level overview of the \korg workflow}
    \label{fig:korg_workflow}
\end{figure}
\vspace{-10pt}

We now describe how the formal gadgets of
Section~\ref{sec:gadgets} are realized within \promela.

\textbf{Channel parsing.}
\korg begins by parsing the user-provided \promela model to
extract all channel declarations, including their names,
buffer sizes, and message types. \korg supports both standard
channels (e.g., \texttt{chan c = [8] of \{int, int, int\}})
and indexed multi-channels (e.g., \texttt{chan c[N] = [8] of \{int\}}).
The parsed channel types determine the shape of the
variables and buffers used in the synthesized gadgets.

\textbf{Drop gadget implementation.}
The drop gadget is synthesized as a \promela \texttt{proctype}
that monitors the target channel using the polling operator
\texttt{?[...]}. When a message is detected and the drop counter
has not been exhausted, the gadget atomically consumes the message
via \texttt{?} and decrements the counter. The $\mathsf{skip}(c)$
action is implemented by tracking \texttt{len(c)} and blocking
until the channel length changes, as described in
Section~\ref{sub:gadget-properties}.

\textbf{Replay gadget implementation.}
The replay gadget is synthesized with a local \promela channel
\texttt{gadget\_mem} serving as the internal FIFO buffer
$\mathsf{Buf}_\ell(M)$. In the $\textsc{Consume}$ phase,
messages are copied from the target channel using the
non-destructive receive operator \texttt{?<...>}
(corresponding to $c\,?\,\langle m \rangle$) and stored into
\texttt{gadget\_mem}. In the $\textsc{Replay}$ phase,
messages are nondeterministically selected from
\texttt{gadget\_mem} using a rotation idiom:
\spin's \texttt{select} statement generates a random index,
and messages are rotated through the FIFO until the selected
message is at the head, at which point it is dequeued and sent
onto the target channel. Messages may also be silently
discarded from the buffer to model partial replay.

\textbf{Reorder gadget implementation.}
The reorder gadget is the most complex. It is synthesized with
\promela's \texttt{priority} mechanism set to the maximum value
(priority 99 in the implementation),
ensuring the gadget always takes precedence over other processes
when consuming messages from the target channel.
This is essential: without priority, other processes may
consume messages from $c$ before the reorder gadget intercepts
them, breaking the reordering semantics. The $\textsc{Init}$
phase uses the $\mathsf{skip}(c)$ implementation to
nondeterministically skip an arbitrary prefix of messages.
The $\textsc{Consume}$ phase uses the destructive receive
\texttt{?} to remove messages from the channel into the local
buffer. The $\textsc{Replay}$ phase uses the same
rotation-based selection as the replay gadget.

\textbf{Gadget injection.}
Once generated, the gadget \promela code is injected into the
user-provided model by appending the generated \texttt{proctype}
declarations. The modified model is then passed to \spin for
verification. Multiple gadgets targeting different channels
can be composed in a single model.

\section{Case Study: TCP}
\label{sec:case-study-tcp}

The Transmission Control Protocol (TCP) is a transport-layer
protocol designed to establish reliable, ordered communication
between two peers. TCP is ubiquitous in today's internet
and has therefore seen ample formal verification efforts,
including using \promela and \spin \cite{conf/secdev/guillaume21,conf/snp/leonor22}.

We constructed a \promela model of TCP referencing the
set of TCP RFCs \cite{rfc793,rfc9293}. Our model follows the state-level control
flow for connection establishment and teardown, the base-level
packet transmission structure, and timeouts.
The model consists of two symmetric \texttt{TCP} processes
communicating over two asynchronous channels,
\texttt{AtoB} and \texttt{BtoA}, each with buffer capacity~2.
Messages are drawn from the domain
$\{\texttt{SYN}, \texttt{FIN}, \texttt{ACK}\}$.

We derive six LTL properties from the RFC specifying
the correctness of our TCP model, and verify that
our model satisfies all six under the assumption of
perfect, fault-free channels:
\begin{itemize}[nosep]
  \item $\phi_1$: No half-open connections (safety).
  \item $\phi_2$: No infinite stalls or deadlocks (liveness).
  \item $\phi_3$: \texttt{SYN\_RECEIVED} eventually resolves to
        \texttt{ESTABLISHED}, \texttt{FIN\_WAIT\_1}, or
        \texttt{CLOSED} (liveness).
  \item $\phi_4$: Simultaneous close eventually resolves (liveness).
  \item $\phi_5$: Active close eventually terminates (liveness).
  \item $\phi_6$: The three-way handshake cannot be bypassed (safety).
\end{itemize}

Crucially, we purposely do not model implementation-level
resilience features such as retransmission timeouts,
duplicate ACK detection, sequence numbers, acknowledgement
numbers, or receive buffers --- features designed precisely
to mitigate channel faults in practice.
We thoroughly assess our limited TCP model against \korg's
drop, replay, and reorder attackers to evaluate their
capabilities in the scenario where the TCP implementation
lacks these resilience features. We evaluated the TCP model against all three attacker models
on individual and combined channels.
The drop and replay limits were set to $\ell = 1$, and the
reorder limit was set to $\ell = 2$.
Our results are summarized in Table~\ref{tab:tcp-results}.

\begin{table}[h]
\vspace{-1.5em}
\centering
\caption{Results of \korg analysis on TCP for $\phi_1$--$\phi_6$.
\checkmark{} indicates \korg proved the absence of an attack
via exhaustive search;
$\times$ indicates an attack was discovered.}
\label{tab:tcp-results}
\smallskip
\begin{tabular}{@{} l ccc @{}}
\toprule
\textbf{Property} & \textbf{Drop} & \textbf{Replay} & \textbf{Reorder} \\
\midrule
$\phi_1$ (half-open prevention)       & $\times$   & $\times$   & \checkmark \\
$\phi_2$ (no deadlocks)               & $\times$   & $\times$   & \checkmark \\
$\phi_3$ (\textsc{SynRec} resolution) & $\times$   & $\times$   & \checkmark \\
$\phi_4$ (simultaneous close)         & $\times$   & $\times$   & \checkmark \\
$\phi_5$ (active close terminates)    & $\times$   & $\times$   & \checkmark \\
$\phi_6$ (handshake bypass)           & \checkmark & \checkmark & \checkmark \\
\bottomrule
\end{tabular}
\vspace{-1.5em}
\end{table}

The drop and replay attackers were able to violate five of the
six properties. For instance, \korg discovered that dropping
A's \texttt{FIN} allows A to eventually time out to
\texttt{Closed} while B remains stranded in \texttt{Established},
violating $\phi_1$. Similarly, replaying a stale \texttt{SYN}
forces the receiver into a spurious handshake state,
also violating $\phi_1$. For the liveness properties
$\phi_2$, $\phi_3$, $\phi_4$, and $\phi_5$,
\korg discovered acceptance cycles in which the attacker
prevents progress by consuming or flooding critical
control messages. The reorder attacker, by contrast, was unable to violate
any property: out-of-order packets are simply skipped by
the TCP state machine, causing endpoints to resolve
to \texttt{Closed} via timeouts rather than entering
an erroneous state.

These results demonstrate the need for the implementation-level
resilience features we intentionally omitted.
In a full TCP implementation, retransmission timeouts and
sequence number checks would mitigate the drop and replay
attacks \korg discovered. \korg's analysis mechanically
confirms this intuition: the base-level state machine alone
is insufficient, and the protocol's correctness depends on
these additional mechanisms.

\section{Related Work}
\label{sec:related}

\textbf{Formal Analysis of Secure Protocols.}
Several formal methods tools reason about protocol security,
primarily in the cryptographic setting, including
Tamarin \cite{tools/tamarin}, ProVerif \cite{tools/proverif},
and CryptoVerif \cite{tools/cryptoverif}.
Model checker-based approaches using \spin \cite{tools/spin} or TLA+ \cite{tools/tla} have
been applied to distributed protocols but have historically
focused on functional correctness rather than channel-level faults.
There is also a long history of using formal methods ad-hoc
to reason about on-channel attackers, particularly in the
context of Byzantine protocols, as well as message tampering,
delays, and congestion control \cite{book/lynch96}.
The automated attacker synthesis framework as presented by von Hippel et al. \cite{conf/safecomp/Hippel20}
was also extended to support stochastic systems by Oakley et al. \cite{conf/csf/oakley22}; however, this extension was designed primarily to target statistical models rather than insecure communication channels as \korg does.

\textbf{Network Testing Tools.}
On the implementation side, Scapy \cite{tools/scapy} provides
programmatic packet crafting and injection,
\texttt{netem} \cite{tools/netem} provides kernel-level fault
injection (loss, delay, reordering), and
packetdrill \cite{tools/packetdrill} enables scripted network
stack testing.
These tools target implementations rather than models and
provide no completeness guarantees.
\korg complements them: it operates on formal models with
sound, complete analysis, while implementation-level tools
remain invaluable for confirming that discovered attacks
manifest in practice.

\section{Conclusion}
\label{sec:conclusion}

We have presented \korg, a tool for the automated, sound,
and complete synthesis of channel fault attacks on distributed
protocols. \korg fills a gap in the formal verification
landscape by providing systematic support for reasoning about
drop, replay, and reorder attackers on communication channels,
operations that were previously modeled only in an ad-hoc
fashion. By leveraging \spin, \korg either synthesizes concrete
attack traces or certifies the absence of attacks via exhaustive
state-space search. We have demonstrated \korg's applicability 
through a case study on TCP, and we release \korg and all models as 
open source.

\bibliographystyle{plain}
\bibliography{main}

\section{Appendix}
\label{sec:appendix}

\subsection{Linear Temporal Logic-B\"uchi Automata Correspondence}
\label{sec:LTL-BA-SEC}
Linear Temporal Logic is traditionally defined under the following semantics:

\begin{definition}[LTL Semantics]
\label{def:ltl-sem}
An LTL formula $\phi$ is interpreted over an infinite word
$\sigma = \sigma_0 \sigma_1 \sigma_2 \cdots \in (2^{AP})^\omega$.
The satisfaction relation $\sigma, i \models \phi$ (``$\phi$
holds at position $i$ in $\sigma$'') is defined inductively:
\begin{align*}
\sigma, i &\models \mathit{true} && \text{always} \\
\sigma, i &\models p && \text{iff } p \in \sigma_i \\
\sigma, i &\models \neg \phi && \text{iff } \sigma, i \not\models \phi \\
\sigma, i &\models \phi_1 \land \phi_2 && \text{iff }
    \sigma, i \models \phi_1 \text{ and } \sigma, i \models \phi_2 \\
\sigma, i &\models \mathbf{X}\,\phi && \text{iff }
    \sigma, i{+}1 \models \phi \\
\sigma, i &\models \phi_1\,\mathbf{U}\,\phi_2 && \text{iff }
    \exists\, j \geq i.\;
    \sigma, j \models \phi_2 \text{ and }
    \forall\, i \leq k < j.\; \sigma, k \models \phi_1
\end{align*}
We write $\sigma \models \phi$ for $\sigma, 0 \models \phi$.
The \textit{language} of $\phi$, denoted $\mathcal{L}(\phi)$,
is the set of all $\sigma \in (2^{AP})^\omega$ such that
$\sigma \models \phi$.
\end{definition}

The fundamental connection between LTL and B\"uchi Automata
is that they characterize the same class of $\omega$-regular
properties:

\begin{theorem}[LTL--BA Correspondence]
\label{thm:ltl-ba}
For every LTL formula $\phi$ over $AP$, there exists a B\"uchi
automaton $B_\phi$ over alphabet $\Sigma = 2^{AP}$ such that
$\mathcal{L}(B_\phi) = \mathcal{L}(\phi)$. The construction
is effective and produces a BA with at most $2^{|\phi|}$ states.
\end{theorem}

This correspondence, due to Vardi and Wolper \cite{conf/lics/Vardi86}, is the basis
of automata-theoretic model checking: to check whether
$P \models \phi$, one constructs $B_{\neg\phi}$ (the BA
accepting exactly the traces that violate $\phi$), computes
the product of the system automaton with $B_{\neg\phi}$,
and checks whether the resulting product has an accepting run.
An accepting run constitutes a counterexample; an empty
language constitutes a proof that $P \models \phi$.

\subsection{Partial Order Reduction}
\label{app:por}

\spin employs partial order reduction (POR) to mitigate
state-space explosion by pruning interleavings that lead
to equivalent states. However, POR can be unsound in the
presence of certain nondeterministic constructs. We compile
all \korg models with the \texttt{-DNOREDUCE} flag, disabling
POR to ensure the exhaustive search remains sound and complete.
While this increases verification time, it is necessary to
preserve the completeness guarantees of \korg.

\subsection{Expanded Proofs and Additional Theorems}
\label{sub:expanded-proofs}

\textbf{Theorem \ref{thm:proc-ba}.} \textit{Every process $P = \langle AP, I, O, S, s_0, T, L \rangle$ directly corresponds to a B\"uchi automaton, and vice versa.}

\begin{proof}
We demonstrate a constructive, bidirectional mapping between the process formalism and B\"uchi Automata (BA) that strictly preserves the language of infinite execution traces.

\textit{Forward Mapping (Process to BA):}
Let $P = \langle AP, I, O, S, s_0, T, L \rangle$ be a process, and let $p \in AP$ be a distinguished atomic proposition representing the acceptance condition. We construct a corresponding BA, $B = (Q, \Sigma, \delta, Q_0, F)$, as follows:
\begin{itemize}
    \item $Q = S$: The states of the automaton are identical to the process states.
    \item $Q_0 = \{s_0\}$: The initial state of the process becomes the singleton initial state set of the BA.
    \item $\Sigma = I \cup O \cup \{\tau\}$: The automaton alphabet is the disjoint union of the process's input, output, and internal actions.
    \item $\delta = T$: The transition relation maps directly. For every $(s, \alpha, s') \in T$, we define $(s, \alpha, s') \in \delta$.
    \item $F = \{s \in S \mid p \in L(s)\}$: The accepting states are exactly those where the proposition $p$ holds.
\end{itemize}
An infinite run in $P$ is a sequence of states $s_0, s_1, \ldots$ produced by transitions in $T$. By definition, this run satisfies the property $\mathbf{G}\mathbf{F} p$ if and only if it visits states where $p \in L(s)$ infinitely often. In the constructed $B$, this corresponds exactly to visiting states in $F$ infinitely often, which is the standard B\"uchi acceptance condition. Thus, the languages are structurally equivalent.

\textit{Backward Mapping (BA to Process):}
Let $B = (Q, \Sigma, \delta, Q_0, F)$ be a B\"uchi automaton. We construct a process $P = \langle AP, I, O, S, s_0^\prime, T^\prime, L \rangle$. Because our process definition requires a single initial state $s_0^\prime$ but a BA may have a set of initial states $Q_0$, we introduce a dedicated starting state if $|Q_0| > 1$.
\begin{itemize}
    \item $S = Q \cup \{s_0^\prime\}$, where $s_0^\prime$ is a fresh state.
    \item $I = \Sigma$ and $O = \emptyset$. All alphabet symbols are treated as input actions.
    \item $T^\prime = \delta \cup \{(s_0^\prime, \tau, q) \mid q \in Q_0\}$: We preserve the BA transitions and add silent $\tau$-transitions from the new initial state $s_0^\prime$ to all original initial states in $Q_0$.
    \item $AP = \{\mathit{accept}\}$.
    \item $L(s) = \{\mathit{accept}\}$ if $s \in F$, and $\emptyset$ otherwise.
\end{itemize}
Any accepting run in $B$ begins in some $q \in Q_0$ and visits $F$ infinitely often. The corresponding run in $P$ begins at $s_0^\prime$, takes a $\tau$-transition to $q$, and proceeds identically, visiting states labeled with $\mathit{accept}$ infinitely often. Therefore, the trace semantics and acceptance properties are preserved in both directions.
\end{proof}

\vspace{1em}

\begin{theorem}
No gadget $D_\ell(c)$, $R_\ell(c)$, or $O_\ell(c)$ admits an infinite run that cycles indefinitely between non-$\textsc{End}$ states while satisfying the labeling condition $\mathit{done}$.
\end{theorem}

\begin{proof}
To prove livelock freedom with respect to the acceptance condition, we must show that no gadget can satisfy the B\"uchi acceptance condition (visiting a state labeled with $\mathit{done}$ infinitely often) without permanently terminating in the $\textsc{End}$ state.

By Definitions \ref{def:drop}, \ref{def:replay}, and \ref{def:reorder}, the labeling function for all three gadgets ($D_\ell(c)$, $R_\ell(c)$, and $O_\ell(c)$) is uniformly defined as:
$$L(s) = \{\mathit{done}\} \iff s = \textsc{End}$$
and $L(s) = \emptyset$ for all $s \neq \textsc{End}$.

Assume, for the sake of contradiction, that there exists an infinite run $\rho = s_0, s_1, s_2, \ldots$ within one of the gadgets that cycles indefinitely exclusively among non-$\textsc{End}$ states, and that this run visits a state labeled with $\mathit{done}$ infinitely often.

Because $\rho$ is confined strictly to non-$\textsc{End}$ states, for all $s_i \in \rho$, $s_i \neq \textsc{End}$. By the definition of $L(s)$, if $s_i \neq \textsc{End}$, then $\mathit{done} \notin L(s_i)$. Therefore, $\mathit{done}$ is never satisfied at any point in the infinite run $\rho$. This directly contradicts the assumption that the run visits a state labeled with $\mathit{done}$ infinitely often.

Furthermore, we examine the internal structure of the gadgets to ensure infinite internal cycling (standard livelock) cannot spuriously mimic progress:
\begin{enumerate}
    \item State-mutating transitions strictly and monotonically decrease finite resources: the consume counter $n$ is decremented upon observing/dropping, and the buffer size $|b|$ is decremented upon replaying. Because $n \leq \ell$ and $|b| \leq \ell$ are bounded, infinite mutation without eventually reaching $n=0$ or $|b|=0$ is impossible.
    \item The non-mutating transitions (\textbf{Pass} $\mathsf{skip}(c)$ self-loops) maintain the current state. An infinite sequence of \textbf{Pass} transitions remains in a non-$\textsc{End}$ state forever, thus never satisfying the $\mathit{done}$ label.
    \item The \textbf{Terminate} or \textbf{Empty} transitions lead strictly to $\textsc{End}$.
\end{enumerate}
Once a gadget transitions to $\textsc{End}$, the only available transition is the $\textsc{End} \xrightarrow{\tau} \textsc{End}$ self-loop. Therefore, to satisfy $\mathit{done}$ infinitely often, the gadget must reach the $\textsc{End}$ state, effectively exhausting its fault injection capabilities. No cyclic livelock satisfying the property exists.
\end{proof}

\begin{theorem}[Composability]
\label{prop:composability}
Multiple gadgets targeting distinct channels compose via
asynchronous parallel composition, preserving soundness
and completeness.
\end{theorem}

\begin{proof}
Each gadget acts on a distinct channel $c_i$, so their action
sets are pairwise disjoint and composition introduces no new
synchronization constraints. By Theorem~\ref{prop:finite},
each gadget is finite-state, so the product remains finite and
\spin exhaustively explores all interleavings, including
coordinated multi-channel attacks.
\end{proof}

\subsection{Case Study: The Alternating Bit Protocol}
\label{sec:case-study-abp}

The alternating bit protocol (ABP) is a simple data-link layer protocol
that provides reliable, in-order delivery over an unreliable channel.
A sender tags each data frame with an alternating sequence bit and
retransmits until it receives a matching acknowledgement; upon
receipt, the receiver flips its expected bit and acknowledges.
ABP is a canonical example of a protocol designed to tolerate
message loss, and has seen formal verification using both
theorem provers and model checkers.

We constructed a \promela model of ABP referencing its
standard description.
The model consists of two processes, a sender and a receiver,
communicating over two asynchronous channels:
\texttt{AtoB} (data) and \texttt{BtoA} (acknowledgements),
each with buffer capacity 1.
Messages are drawn from the domain
$\{\texttt{DATA0}, \texttt{DATA1}, \texttt{ACK0}, \texttt{ACK1}\}$.
We formalize the following liveness property:
\[
  \phi_{\mathit{ABP}} = \mathbf{G}\,\bigl((\mathit{packetA} \neq \mathit{packetB})
  \implies \mathbf{F}\,(\mathit{packetA} = \mathit{packetB})\bigr)
\]
\noindent That is, whenever the sender and receiver disagree on the
current data frame, they eventually reach agreement.
Our \promela model satisfies $\phi_{\mathit{ABP}}$ in the absence of
any attacker.

We evaluated the ABP model against \korg's drop, replay, and
reorder attacker models on both individual and combined channels.
The drop and replay limits were set to $\ell = 2$, and the reorder
limit was set to $\ell = 2$. Our results are summarized in
Table~\ref{tab:abp-results}.

\begin{table}[h]
\vspace{-1em}
\centering
\caption{Results of \korg analysis on ABP for $\phi_{\mathit{ABP}}$.
\checkmark{} indicates \korg proved the absence of an attack
via exhaustive state-space search.}
\label{tab:abp-results}
\smallskip
\begin{tabular}{@{} l ccc @{}}
\toprule
\textbf{Victim Channel(s)} & \textbf{Drop} & \textbf{Replay} & \textbf{Reorder} \\
\midrule
\texttt{AtoB}                & \checkmark & \checkmark & \checkmark \\
\texttt{BtoA}                & \checkmark & \checkmark & \checkmark \\
\texttt{AtoB} + \texttt{BtoA} & \checkmark & \checkmark & \checkmark \\
\bottomrule
\end{tabular}
\vspace{-1em}
\end{table}

In all nine configurations, \korg completed an exhaustive
state-space search and certified the absence of attacks.
This result is expected: ABP's retransmission mechanism
ensures that dropped messages are resent, replayed messages
are discarded by the sequence-bit check, and reordered
messages (within a single-capacity buffer) do not affect
correctness. \korg's analysis confirms mechanically what is
known informally about ABP's resilience to channel faults,
and serves as a validation of \korg's gadget constructions
on a well-understood protocol.

\end{document}